\documentclass[reqno,12pt,centertags]{amsart}
\usepackage{amsmath,amsthm,amscd,amssymb,latexsym,upref,stmaryrd}
\usepackage[cp1251]{inputenc}
\usepackage[english]{babel}
\usepackage{times}
\usepackage{amsfonts}
\usepackage[numbers,sort&compress]{natbib}
\usepackage{hyperref}
\newcommand*{\mailto}[1]{\href{mailto:#1}{\nolinkurl{#1}}}
\usepackage{amsmath}
\usepackage{amsthm}

\newtheorem{theorem}{Theorem}[section]
\newtheorem{lemma}{Lemma}[section]
\newtheorem{remark}{Remark}[section]
\newtheorem{prop}{Proposition}[section]

\numberwithin{equation}{section}
\begin{document}

\thispagestyle{empty}

\noindent{\large\bf On the direct and inverse transmission eigenvalue problems for the Schr\"{o}dinger operator on the half line}
\\

\noindent {\bf  Xiao-Chuan Xu}\footnote{School of Mathematics and Statistics, Nanjing University of Information Science and Technology, Nanjing, 210044, Jiangsu,
People's Republic of China, {\it Email:
xcxu@nuist.edu.cn}}
\\

\noindent{\bf Abstract.}
{For the direct problem, we give the asymptotic distribution of the (real and non-real) transmission eigenvalues for the Schr\"{o}dinger operator on the half line. For the inverse problem, we prove that the potential can be uniquely determined by all transmission eigenvalues, the parameter in the boundary condition and some non-zero value related to the potential (whereas without the certain constant $\gamma$). The reconstruction algorithms are also revealed.}



\section{Introduction}
Consider the following problem $R(q,h):$
\begin{align}
        &-\psi''+q(x)\psi =\lambda\psi, \quad 0<x<1,\label{1}\\
             &- \psi_0''=\lambda\psi_0,\;\qquad\qquad 0<x<1,\label{2}\\
             &\psi'(0)-h\psi(0)=0,\label{3}\\
             &\psi_0'(0)-h\psi_0(0)=0,\label{4}\\
             &\psi_0(1)=\psi(1),\quad \psi_0'(1)=\psi'(1),\label{5}
 \end{align}
where the potential $q(x)$ is real and belongs to $ L^2(0,1)$, $\lambda$ is the spectral parameter, and $h\in \mathbb{R}$.
The $\lambda$-values for which the problem $R(q,h)$ has a  pair of nontrivial solutions $\{\psi,\psi_0\}$  are called \emph{transmission eigenvalues}, which are the energies at which the scattering from the 'perturbed' system agrees with the scattering from the 'unperturbed' system \cite{TA1}.

 Recently, the transmission eigenvalue problems have attracted wide attention  (see, e.g., \cite{TA1,TA2,TA4,BB,SA,SA1,CCG,DC2,DC3,CL,JR,JR1,WS,WX,XC0,XC,XY2,XY3,YC,YB} and the references therein).
 There are many interesting studies for the transmission eigenvalue problem with Dirichlet boundary condition at $x=0$ (i.e., $\psi(0)=0=\psi_0(0)$ instead of (\ref{3}) and (\ref{4}), respectively). However, when it comes to the boundary condition (\ref{3}), only a few has been done \cite{TA1,XY3}. Note that there are important physical problems where the boundary condition (\ref{3}) rather than Dirichlet boundary condition at $x=0$ is appropriate
to use. For example, the inverse scattering problem for determining  the shape of the human  vocal tract \cite{TA3}.

The problem $R(q,h)$ was first considered by Aktosun and Papanicolaou \cite{TA1}, who showed that
 the potential is uniquely determined by the set of all the transmission eigenvalues, the parameter $h$ in the boundary condition and the additional constant $\gamma$ appearing in (\ref{dk}).  In the same paper, they raised an open  question that whether the constant $\gamma$ and the parameter $h$ may be contained in knowledge of transmission eigenvalues. It was shown in \cite{XY3}
 that  the uniqueness will not hold without $\gamma$. Since the constant $\gamma$ has no physical meaning and may not be measured, we should try to find some other information to replace it. In the Dirichlet case, Wei and Xu \cite{WX} suggested to use some knowledge of the potential near $x=1$ to replace $\gamma$. They showed that if $q\in C^{(m)}(1-\delta,1]$ for some $\delta\in(0,1)$, and $q^{(j)}(1)=0$ for $j=\overline{0,m-1}$ and $q^{(m)}(1)\ne0$ is known, then the uniqueness holds without $\gamma$. Later, Xu and Yang \cite{XC0}  provided a reconstruction algorithm corresponding to the uniqueness theorem in \cite{WX} for the case $m=0$. There is still no result for the reconstruction algorithm for the cases $m\ge1$.

 In this paper, we will generalize the result of Wei and Xu into the non-Dirichlet case, by using a completely different but briefer method. The proof in this paper implies a reconstruction algorithm for recovering $q(x)$ from all transmission eigenvalues, $h$ and the non-zero $q^{(m)}(1)$. This reconstruction algorithm can also  be applied to the Dirichlet case.

Note that the problem $R(q,h)$ is non-self-adjoint. Therefore, it is possible that non-real transmission eigenvalues exist. The existence and distribution of the nonreal transmission eigenvalues for the Dirichlet case have been studied in \cite{DC2,DC3,WS,XC,XY2}. In the current paper, we show that in some cases, there exist infinitely many real eigenvalues and at most a finite number of nonreal eigenvalues, and in some other cases, there are infinitely many nonreal eigenvalues and at most a finite number of real eigenvalues. Moreover, we give the asymptotic behavior of the transmission eigenvalues (including real and nonreal). In the Dirichlet case \cite{XY2}, we gave the asymptotics of the non-real eigenvalues without the description of the subscript numbers, which will be supplemented in this paper. The asymptotics of eigenvalues with the description of subscript numbers has important applications in the inverse spectral problems, especially for the stability and reconstruction algorithm (see, e.g., \cite{BB,BN,FY1,RS,RS1}).

The paper is organized as follows. In Section 2, we study the asymptotics of the characteristic function, and give the asymptotic solution of a transcendental equation. In Sections 3, we study the asymptotics of the zeros of the leading function. In Section 4, the asymptotics for the transmission eigenvalues are investigated. In Section 5, we consider the inverse spectral problem for the problem $R(q,h)$. In the last section, some results related to the Dirichlet case are provided.

\section{Preliminaries}

In this section, we provide some properties  the Jost solution of the equation (\ref{1}), which is used to obtain the asymptotics of the characteristic function, and give the known knowledge for seeking the asymptotics of the non-real eigenvalues.

Let $\lambda=k^2$, and $f(k,x)$ be the Jost solution to (\ref{1}), which satisfies $f(k,1)\!=e^{ik}$ and  $f'(k,1)=ike^{ik}$. It is known  \cite{VM} that
\begin{equation}\label{4s}
  f(k,x)=e^{ikx}+\int_x^{2-x}K(x,t)e^{ikt}dt,
\end{equation}
where $K(x,t)$ satisfies
\begin{equation}\label{xx}
  K(x,t)=\frac{1}{2}\int_{\frac{x+t}{2}}^1 q(s)ds+\frac{1}{2}\int_x^1 q(s)\int_{t-(s-x)}^{t+(s-x)}K(s,u)duds,
\end{equation}
\begin{equation}\label{xx1}
  K(x,t)=0\quad \text{if $x+t\ge2$ or $x>t$}.
\end{equation}
From (\ref{xx}) we see that
\begin{equation}\label{xx2}
  K_x(x,t)=-\frac{1}{4}q\left(\frac{x+t}{2}\right)+\frac{1}{2}\int_x^1q(s)[-K(s,t+s-x)-K(s,t-s+x)]ds,
\end{equation}
\begin{equation}\label{xx3}
  K_t(x,t)=-\frac{1}{4}q\left(\frac{x+t}{2}\right)+\frac{1}{2}\int_x^1q(s)[K(s,t+s-x)-K(s,t-s+x)]ds.
\end{equation}
It follows from (\ref{xx1}),  (\ref{xx2}) and (\ref{xx3}) that
\begin{equation}\label{xx4}
  K_x(0,2)= K_t(0,2)=-\frac{1}{4}q\left(1\right), \quad K_t(x,-x)=-\frac{1}{4}q(0),
\end{equation}
\begin{equation}\label{xxs}
  K_t(x,2-x)= -\frac{1}{4}q\left(1\right), \quad K_t(x,2+x)= 0, \quad x>0,
\end{equation}
\begin{equation}\label{xx5}
  K_x(x,x)\!= -\frac{1}{4}q\left(x\right)-\frac{1}{8}\!\left[\int_x^1\!q(s)ds\right]^2, \quad K_t(x,x)\!=\! -\frac{1}{4}q\left(x\right)+\frac{1}{8}\!\left[\int_x^1\!q(s)ds\right]^2.
\end{equation}
Here we have used the formula
\begin{equation*}
  \int_{t_n}^1\!\int_{t_{n-1}}^1\!\cdot\cdot\cdot \!\int_{t_1}^1\!\left[\prod_{j=0}^{n-1}q(t_j)\right]dt_0\cdot\cdot\cdot dt_{n-1}=\frac{1}{n!}\!\left(\!\int_{t_n}^1q(s)ds\!\right)^n,\quad t_n\le t_{n-1}\le\cdot\cdot\cdot\le t_0.
\end{equation*}
Denote
\begin{equation}\label{5s}
F(k):=-i[f'(k,0)-hf(k,0)].
\end{equation}
It was shown in \cite{TA1} that the transmission eigenvalues of the problem $R(q,h)$ coincide with the  zeros of the  \emph{characteristic function} $\Delta(\lambda):=D(k)$, where
\begin{equation}\label{6s}
  D(k)=\frac{1}{2i}[F(k)+F(-k)]-\frac{h}{2k}[F(k)-F(-k)].
\end{equation}
It is obvious that $D(k)$ is an even and entire function of $k$ of exponential type. By the Hadamard's factorization theorem, we have
\begin{equation}\label{dk}
  D(k)=\gamma E(k),\quad E(k):=k^{2s}\prod_{\lambda_n\ne0}\left(1-\frac{k^2}{\lambda_n}\right),
\end{equation}
where $\gamma$ is constant, $s\ge0$ is the multiplicity of the zero eigenvalue, and $\{\lambda_n\}$ are the nonzero transmission eigenvalues.
Moreover, the function $D(k)$ satisfies (see \cite{TA1})
$$D(k)^*=D(-k^*),\quad \forall k\in \mathbb{C},$$
 where $k^*$ means the conjugate of $k$. Thus, the distribution of the zeros of function $D(k)$   is symmetrical with respect to both the real and imaginary axes.

\begin{prop}
The above function $D(k)$ can be expressed as
\begin{align}\label{p2p}
\notag D(k)=&\frac{\omega}{2}-\int_0^2K_x(0,t)\cos ktdt+\frac{h^2\omega}{2k^2}\\
&- \frac{h}{k}\int_0^2[K_t(0,t)+K_x(0,t)]\sin ktdt+\frac{h^2}{k^2}\int_0^2 K_t(0,t)\cos kt dt,
\end{align}
where $K(x,t)$ satisfies (\ref{xx}), and
\begin{equation}\label{3ss}
\omega:=\int_0^1 q(s)ds.
\end{equation}
\end{prop}

\begin{proof}
From  (\ref{6s}) we have
\begin{align}
\notag  D(k)=&-\frac{1}{2}[(f'(k,0)+f'(-k,0))-h(f(k,0)+f(-k,0))]\\
  &-\frac{h}{2ik}[[(f'(k,0)-f'(-k,0))-h(f(k,0)-f(-k,0))].\label{7s}
\end{align}
By virtue of (\ref{4s}) and (\ref{xx1}), we get
\begin{align}
 \notag f'(k,0)+f'(-k,0)=&ik-\!K(0,0)\!+\!\int_0^2 \!K_x(0,t)e^{ikt}dt\\
 \notag &\quad -\!ik\!-\!K(0,0)\!+\!\int_0^2\! K_x(0,t)e^{-ikt}dt\\
 =& -2K(0,0)+2\int_0^2K_x(0,t)\cos ktdt.\label{8s}
\end{align}
Similarly, we can also obtain
\begin{equation}\label{9s}
 f'(k,0)-f'(-k,0) =2ik+2i\int_0^2K_x(0,t)\sin ktdt,
\end{equation}
\begin{equation}\label{10s}
 f(k,0)+f(-k,0) =2+{2}\int_0^2K(0,t)\cos ktdt,
\end{equation}
 \begin{equation}\label{11s}
 f(k,0)-f(-k,0) ={2i}\int_0^2K(0,t)\sin ktdt.
\end{equation}
Substituting (\ref{8s})-(\ref{11s}) into (\ref{7s}), we get
\begin{align}\label{12s}
\notag D(k)\!=\!K(0,0)\!&+\!\int_0^2 \![hK(0,t)-K_x(0,t)]\cos kt dt\\
 &+\frac{h}{k}\int_0^2 [hK(0,t)-K_x(0,t)]\sin kt dt.
\end{align}
Then integrating by parts in (\ref{12s}), and using (\ref{xx}) and (\ref{xx1}), we arrive at (\ref{p2p}).
\end{proof}

If $q\in W_2^1[0,1]$, then $K_x(0,\cdot),K_t(0,\cdot)\in W_2^1[0,1]$. Integrating by parts in (\ref{p2p}), and with the help of (\ref{xx4}) and (\ref{xx5}), we have
\begin{align}\label{t2}
\!\!\!\notag D(k)=&\frac{\omega}{2}\!+\!\frac{q(1)\sin 2k}{4k}\!+\!\frac{1}{k}\int_0^2K_{xt}(0,t)\sin ktdt\!+\!\frac{h}{2k^2}[q(0)-q(1)\cos 2k+\omega h]\\
&-\frac{h}{k^2}\!\!\int_0^2[K_{tt}(0,t)+K_{xt}(0,t)]\cos ktdt+\frac{h^2}{k^2}\int_0^2 K_t(0,t)\cos kt dt.
\end{align}

If $q\in W_2^2[0,1]$, taking the derivative on both sides of (\ref{xx3}) with respect to $t$, and letting $x=0$, we have
\begin{equation}\label{t5}
  K_{xt}(0,t)=-\frac{1}{8}q'\left(\frac{t}{2}\right)+\frac{1}{2}\int_0^1q(s)[-K_t(s,s+t)-K_t(s,t-s)]ds,
\end{equation}
which implies from (\ref{xxs}) that
\begin{align}\label{t6}
 K_{xt}(0,2)=-\frac{1}{8}q'(1)+\frac{q(1)\omega}{8},
\end{align}
 and from (\ref{xx1}), (\ref{xx4}) and (\ref{xx5}) that
\begin{align}\label{t7}
 \notag K_{xt}(0,0)&=-\frac{1}{8}q'(0)+\frac{1}{2}\int_0^1q(s)\left[\frac{q(s)}{4}+\frac{q(0)}{4}-\frac{1}{8}\left(\int_s^1q(u)du\right)^2\right]ds\\
 &=-\frac{1}{8}q'(0)+\frac{1}{8}\int_0^1q^2(s)ds+\frac{q(0)\omega}{8}-\frac{\omega^3}{48}.
\end{align}
Integrating by parts in (\ref{t2}), and using (\ref{t6}) and (\ref{t7}), we have
\begin{align}\label{t8}
\notag D(k)=&\frac{\omega}{2}+\frac{q(1)\sin 2k}{4k}+\frac{Q_1\cos 2k}{8k^2}+\frac{Q_2}{8k^2}\\
&+\frac{1}{k^2}\int_0^2[K_{xtt}(0,t)-h(K_{tt}(0,t)+K_{xt}(0,t))+{h^2}K_t(0,t)]\cos kt dt,
\end{align}
where
\begin{equation}\label{t9}
 Q_1\!=\!q'(1)-q(1)\omega-4hq(1),\;Q_2\!=-q'(0)+\!\int_0^1\!q^2(s)ds+q(0)\omega-\frac{\omega^3}{6}+4h[q(0)+\omega h].
\end{equation}

To get the asymptotics of the non-real transmission eigenvalues, we introduce the following transcendental equation
\begin{equation}\label{k5}
  z-\kappa\log z=w,
\end{equation}
where $\kappa$ is a constant in $\mathbb{C}$ and $\log z=\log|z|+i\arg z$ with $-\pi<\arg z\le \pi$.
\begin{prop}\label{p3.1}
The transcendental equation (\ref{k5}) has a unique solution
\begin{equation}\label{k6}
  z(w)=w+\kappa\log w+\kappa^2\frac{\log w}{w}+O\left(\frac{\log^2 |w|}{|w|^2}\right)
\end{equation}
for any sufficiently large $|w|$.
\end{prop}
\begin{proof}
The uniqueness has been proved in \cite{XY2} (see also \cite{MF}). Let us show (\ref{k6}). By the Appendix  in \cite{XY2}, we see that the solution $z$ of (\ref{k5}) satisfies
\begin{equation}\label{ap2}
  z=w(1+\xi),
\end{equation}
where $\xi$ satisfies the equation
\begin{equation}\label{ap3}
  \xi=\kappa\left(\frac{\log w}{w}+\frac{\log(1+\xi)}{w}\right).
\end{equation}
Note that
\begin{equation}\label{ap4}
  \log(1+\xi)=\xi+O(\xi^2),\quad \xi\to0\;\;(i.e., |w|\to\infty).
\end{equation}
Substituting (\ref{ap3}) and (\ref{ap4}) into (\ref{ap2}), we obtain (\ref{k6}).
\end{proof}

\section{Asymptotics of zeros of the leading function}
In this section, we study the asymptotics of the zeros of the function $g_1(k)$ defined in (\ref{3.1}) below.   In this and the next sections, when we count zeros of a function, we always count the multiple zeros  (if there exist) with their multiplicities.

 Denote
\begin{equation}\label{3.1}
g_1(k):=8ik\left(\frac{\omega}{2}+\frac{q(1)\sin2k}{4k}\right)= 4ik\omega+q(1)(e^{2ik}-e^{-2ik}).
\end{equation}
It is obvious that if $\omega=0$ then the zeros of $g_1(k)$ are $\{\frac{n\pi}{2}\}_{n\in \mathbb{Z}}$.
If $\omega\ne0$ then there are infinitely many non-real zeros of $g_1(k)$.

\begin{lemma}\label{lemma3.0}
If $q(1)\ne0$ then all zeros of $g_1(k)$, excepting at most two real or imaginary ones, are simple algebraically.
\end{lemma}
\begin{proof}
If $\omega=0$ it is obvious. Assume $\omega\ne0$, and let $\mu_0$ be some zero of $g_1$ with multiplicities at lest two. It is obvious that  $\mu_0=0$ if and only if $\omega=-q(1)$.  Assume $\mu_0=\sigma_0+i\tau_0$ with $\sigma_0^2+\tau_0^2\ne0$. Then we have from (\ref{3.1}) that
  \begin{equation}\label{zp0}
    \sin(2\mu_0)=-\frac{2\omega \mu_0}{q(1)},\quad    \cos(2\mu_0)=-\frac{\omega}{q(1)},
  \end{equation}
which implies
\begin{equation}\label{zp1}
  \frac{\omega^2}{q(1)^2}(1+4\mu_0^2)=1.
\end{equation}
It follows that $\mu_0^2\in\mathbb{R}$, and so either $\sigma_0=0$ or $\tau=0$. If $\tau=0$ then
\begin{equation}\label{zp2}
  \mu_0=\sigma_0=\frac{\pm\sqrt{\frac{q(1)^2}{\omega^2}-1}}{2}\quad \text{and}\quad \left|\frac{\omega}{q(1)}\right|<1.
\end{equation}
If $\sigma_0=0$ then
\begin{equation}\label{zp3}
  \mu_0=i\tau_0=\frac{\pm i\sqrt{1-\frac{q(1)^2}{\omega^2}}}{2}\quad \text{and}\quad \frac{\omega}{q(1)}<-1.
\end{equation}
The equations (\ref{zp0})-(\ref{zp2}) or (\ref{zp0}), (\ref{zp1}) and (\ref{zp3}) are overdetermined, which may hold only for certain $\frac{\omega}{q(1)}$.  If $\frac{\omega}{q(1)}>1$, then all zeros of $g_1(k)$ are simple. The proof is complete.
\end{proof}
Now, let us study the asymptotics of the zeros of $g_1(k)$. Let $z:=-2ik$ in  (\ref{3.1}) with ${\rm Im}k\ge0$. Then we have that $e^{2ik}g_1(k)=0$ for ${\rm Im}k>0$ is equivalent to that $$ze^{-z}=-\frac{q(1)}{2w}(1-e^{-2z}),$$
which implies
\begin{equation}\label{ap5}
z-\log z=w_n,\quad w_n:=-2n\pi i-\log \left(-\frac{q(1)}{2\omega}\right)-\log \left(1-e^{-2z}\right).
\end{equation}
By Prop. \ref{p3.1} we have
\begin{equation}\label{ap6}
  z=w_n+\log w_n+\frac{\log w_n}{w_n} +O\left(\frac{\log^2 |w_n|}{|w_n|^2}\right),\quad n\to\pm\infty.
\end{equation}
It follows that $ {\rm Re}z= \log n+O(1)$ which implies
\begin{equation}\label{ap5s0}
 \log \left(1-e^{-2z}\right)=O\left(\frac{1}{n^2}\right).
\end{equation}
Going back to the second equation in (\ref{ap5}), we have
\begin{align}\label{ap5s1}
\notag \log w_n=&\log \left[(-2n\pi i)\left(1+\frac{\log \left(-\frac{q(1)}{2\omega}\right)}{2n\pi i}+O\left(\frac{1}{n^3}\right)\right)\right]\\
=&\log (-2n\pi i)+\frac{\log \left(-\frac{q(1)}{2\omega}\right)}{2n\pi i}+O\left(\frac{1}{n^2}\right),
  \end{align}
  and
  \begin{equation}\label{ap5s}
    \frac{\log w_n}{w_n}=\frac{\log (-2n\pi i)}{-2n\pi i}+O\left(\frac{\log n}{n^2}\right).
  \end{equation}
Therefore, substituting (\ref{ap5s0})- (\ref{ap5s}) into  (\ref{ap5}) we have
\begin{equation}\label{ap7}
 \!z\!=\!-2n\pi i +\log (-2n\pi i)-\log \left(-\frac{q(1)}{2\omega}\right)+\frac{\log (-2n\pi i)}{-2n\pi i}+\frac{\log \left(-\frac{q(1)}{2\omega}\right)}{2n\pi i}+O\left(\frac{\log^2 n}{n^2}\right),
\end{equation}
which implies from $z=-2ik$ that the zeros of $g_1(k)$ in $\mathbb{C}_+:=\{k\in \mathbb{C}:{\rm Im}k>0\}$, denoted by $\{\mu_n\}$, have the following asymptotics
\begin{equation}\label{ap8}
 \mu_n=n\pi+\frac{ib_n}{2}-\frac{b_n}{4n\pi}+O\left(\frac{\log^2 n}{n^2}\right),\quad  n\to\pm\infty,
\end{equation}
where
\begin{equation}\label{ap9}
b_n=\log (-2n\pi i)-\log \left(-\frac{q(1)}{2\omega}\right).
\end{equation}
Since the distribution of the zeros of function $g_1(k)$   is symmetrical with respect to the real and imaginary axes.
 Let us consider the zeros of $g_1(k)$ in the
first quadrant, namely, assume $n>0$. In this case, we have
\begin{equation}\label{ap10}
  b_n=\left\{\begin{split}
            & \log (2n\pi)-\log \left(-\frac{q(1)}{2\omega}\right)-\frac{\pi i}{2},\quad \text{if}\quad \frac{q(1)}{\omega}<0,\\
            &\log (2n\pi)-\log \left(\frac{q(1)}{2\omega}\right)-\frac{3\pi i}{2},\quad \text{if}\quad \frac{q(1)}{\omega}>0.
             \end{split}\right.
\end{equation}
Let $\mu_n:=\sigma_n+i\tau_n$ with $\sigma_n>0$ and $\tau_n>0$. With the help of (\ref{ap8}) and (\ref{ap10}), we have that if  $\frac{q(1)}{\omega}<0$ then
\begin{equation}\label{ap11}
  \left\{\begin{split}
            \sigma_n=&\left(n+\frac{1}{4}\right)\pi -\frac{ \log (2n\pi)-\log \left(-\frac{q(1)}{2\omega}\right)}{4n\pi}+O\left(\frac{\log^2 n}{n^2}\right),\\
            \tau_n=&\frac{1}{2}\left[\log (2n\pi)-\log \left(-\frac{q(1)}{2\omega}\right)+\frac{1}{4n}\right]+O\left(\frac{\log^2 n}{n^2}\right);
             \end{split}\right.
\end{equation}
if  $\frac{q(1)}{\omega}>0$ then
\begin{equation}\label{ap12}
  \left\{\begin{split}
            \sigma_n=&\left(n+\frac{3}{4}\right)\pi -\frac{ \log (2n\pi)-\log \left(\frac{q(1)}{2\omega}\right)}{4n\pi}+O\left(\frac{\log^2 n}{n^2}\right),\\
            \tau_n=&\frac{1}{2}\left[\log (2n\pi)-\log \left(\frac{q(1)}{2\omega}\right)+\frac{3}{4n}\right]+O\left(\frac{\log^2 n}{n^2}\right).
             \end{split}\right.
\end{equation}

Let $k=\sigma+i\tau$. Rewrite (\ref{3.1}    ) as
\begin{align}\label{qoa}
 g_1(k)=&-4\omega\tau+q(1)(e^{-2\tau}-e^{2\tau})\cos 2\sigma +i[4\omega\sigma+q(1)(e^{-2\tau}+e^{2\tau})\sin 2\sigma].
\end{align}
For sufficiently large $n\in \mathbb{N}$, consider the rectangle contour $\Gamma_n=I_1\cup I_2\cup I_3 \cup I_4$ (see Figure 1), where
\begin{equation*}
\begin{split}
I_1:&=\left\{k:\sigma=\left(n+1\right)\pi,|\tau|\leq\left(n+1\right)\pi\right\},\;\; I_2:=\left\{k:|\sigma|\le\left(n+1\right)\pi=\tau\right\},\\
I_3:&=\left\{k:-\sigma=\left(n+1\right)\pi\ge|\tau|\right\},\;\;\; I_4:=\left\{k:|\sigma|\le\left(n+1\right)\pi=-\tau\right\}.\\
\end{split}
\end{equation*}
\begin{center}
\setlength{\unitlength}{1mm}
\begin{picture}(60,90)(-30,-45)
\put(0,0){\vector(1,0){55}}
\put(0,0){\line(-1,0){55}}
\put(0,0){\line(0,-1){39}}
\put(0,0){\vector(0,1){43}}
\put(0,30){\line(-1,0){30}}\put(0,30){\line(1,0){30}}
\put(30,0){\line(0,1){30}}\put(30,0){\line(0,-1){30}}
\put(0,-30){\line(-1,0){30}}\put(0,-30){\line(1,0){30}}
\put(-30,0){\line(0,1){30}}\put(-30,0){\line(0,-1){30}}
\put(31,-4){$(n+1)\pi$} \put(30,0){\circle*{1}}
\put(-49,-4){$-(n+1)\pi$} \put(-30,0){\circle*{1}}
\put(1,32){$(n+1)\pi i$} \put(0,30){\circle*{1}}
\put(1,-35){$-(n+1)\pi i$} \put(0,-30){\circle*{1}}
\put(30,2){\vector(0,1){10}}
\put(-30,-2){\vector(0,-1){10}}
\put(2,-30){\vector(1,0){10}}
\put(2,30){\vector(-1,0){10}}
\put(32,5){$I_1$}
\put(-12,32){$I_2$}
\put(-36,4){$I_3$}
\put(-14,-34){$I_4$}
\put(58,-1){$\sigma$}
\put(2,40){$i\tau$}
 \put(-3,-4){0}\put(0,0){\circle*{1}}
\put(-20,-47){Figure 1. The contour $\Gamma_n$}
\end{picture}
\end{center}
\vspace{5mm}
Letting $k$ start from the point $((n+1)\pi,-(n+1)\pi i)$ and travel round $\Gamma_n$ by the counter-clockwise, and using (\ref{qoa}), one can easily obtain that if  $\frac{q(1)}{\omega}<0$ then
 the variations
  \begin{equation*}
  \begin{split}
   \Delta_{I_1}\!\arg\frac{ g_1(k)}{\omega}\!= -\theta_1,\quad \Delta_{I_2}\!\arg \frac{ g_1(k)}{\omega}\!=(4n+4)\pi-\theta_2,\\
   \!\Delta_{I_3}\!\arg \frac{ g_1(k)}{\omega}=-\theta_3,\quad  \Delta_{I_4}\!\arg \frac{ g_1(k)}{\omega}=(4n+4)\pi-\theta_4;
  \end{split}
 \end{equation*}
 if $\frac{q(1)}{\omega}>0$ then the variations
   \begin{equation*}
  \begin{split}
   \Delta_{I_1}\!\arg \frac{ g_1(k)}{\omega}\!= \theta_1,\quad \Delta_{I_2}\!\arg \frac{ g_1(k)}{\omega}\!=(4n+4)\pi+\theta_2,\\
   \!\Delta_{I_3}\!\arg \frac{ g_1(k)}{\omega}=\theta_3,\quad  \Delta_{I_4}\!\arg \frac{ g_1(k)}{\omega}=(4n+4)\pi+\theta_4;
  \end{split}
 \end{equation*}
where $\theta_i\in(0,\pi)$ and
 $\theta_1+\theta_2+\theta_3+\theta_4=2\pi$. Thus $\Delta_{\Gamma_n}\!\arg \frac{ g_1(k)}{\omega}=(8n+6)\pi$ if $\frac{q(1)}{\omega}<0$; $\Delta_{\Gamma_n}\!\arg \frac{ g_1(k)}{\omega}=(8n+10)\pi$ if $\frac{q(1)}{\omega}>0$ .
By the argument principal of entire functions, the number of zeros of the function $g_1(k)$ inside $\Gamma_n$ equals $4n+3$ if $\frac{q(1)}{\omega}<0$; equals $4n+5$ if $\frac{q(1)}{\omega}>0$. Therefore, the zeros of $g_1(k)$ can be numbered as follows.
\begin{lemma}\label{lemma3.1}
If $\frac{q(1)}{\omega}<0$, then the zeros of $\frac{g_1(k)}{k}$, denoted by $\{\pm\mu_n\}_{n\ge0}\cup \{\pm\mu_n^*\}_{n\ge1}$ with ${\rm Im}\mu_n\ge0$ and ${\rm Re}\mu_n\ge0$, satisfy the asymptotic formula (\ref{ap11}). If $\frac{q(1)}{\omega}>0$, then the zeros of $\frac{g_1(k)}{k}$, denoted by $\{\pm\mu_n\}_{n\ge0}\cup \{\pm\mu_n^*\}_{n\ge0}$ with ${\rm Im}\mu_n\ge0$ and ${\rm Re}\mu_n\ge0$, satisfy the asymptotic formula (\ref{ap12}).
\end{lemma}

\section{Asymptotics of the transmission eigenvalues}

In this section, we study the asymptotics of the transmission eigenvalues under the assumption $q\in W_2^1[0,1]$ or $q\in W_2^2[0,1]$. In the latter case, we obtain the more accurate estimate for the asymptotics of the transmission eigenvalues.

Using (\ref{t2}) and (\ref{3.1}), and taking (\ref{12s}) into account, we can rewrite (\ref{t2}) as
\begin{equation}\label{uo1}
8ikD(k)=g_1(k)+\int_{0}^2 K(t)\sin ktdt,
\end{equation}
where $K(t)$ is some function in $L^2(0,2)$.
For arbitrary small $\varepsilon>0$ and large $n$, consider the rectangle contour $\gamma_n:=\gamma_{\sigma_n}^\pm\cup\gamma_{\tau_n}^\pm$, where
  \begin{equation*}
    \gamma_{\sigma_n}^{\pm}:=\{k:\sigma=\sigma_n\pm\varepsilon,|\tau-\tau_n|\leq\varepsilon\},\quad
    \gamma_{\tau_n}^{\pm}:=\{k:\tau=\tau_n\pm\varepsilon,|\sigma-\sigma_n|\leq\varepsilon\}.
  \end{equation*}
  \begin{lemma}\label{lemma3.2}
If $q(1)\ne0$, then    for sufficiently large $n>0$ and small $\varepsilon>0$, the function $g_1(k)$ satisfies
    \begin{equation}\label{ap15}
     | g_1(k)|\ge C_\varepsilon e^{2|{\rm Im}k|},\quad \forall k\in \Gamma_n\cup\gamma_n,
    \end{equation}
    where $C_\varepsilon>0$ depends only on $\varepsilon$.
  \end{lemma}
\begin{proof}
For $\omega=0$ it was proved in \cite[p.10]{FY1}. Assume $\omega\ne0$.
Let us first consider the case $k\in \gamma_n$.
Denote $G_1(k):=|g_1(k)|^2e^{-4|{\rm Im}k|}$, and
\begin{equation*}
  c_n:=\inf_{k\in \gamma_n}G_1(k).
\end{equation*}
It obvious that $c_n>0$ for all $n\ge0$. Let us show that the sequence $\{c_n\}_{n\ge0}$ has a positive lower bound $c_\varepsilon$ which is independent of $n$. To this end,
recall $k=\sigma+i\tau$. By a direct calculation, we have
\begin{align}\label{ap16}
 \notag |g_1(k)|^2= &16\omega^2(\sigma^2+\tau^2)+q(1)^2(e^{-4\tau}+e^{4\tau})-2q(1)^2\cos4\sigma\\
 &+8\omega q(1)[(e^{-2\tau}+e^{2\tau})\sigma\sin 2\sigma-(e^{-2\tau}-e^{2\tau})\tau\cos2\sigma].
\end{align}
It follows that
\begin{align}\label{ap16}
 \notag G_1(k)= &16\omega^2(\sigma^2+\tau^2)e^{-4\tau}+q(1)^2(e^{-8\tau}+1)-2q(1)^2e^{-4\tau}\cos4\sigma\\
 &+8\omega q(1)[(e^{-6\tau}+e^{-2\tau})\sigma\sin 2\sigma-(e^{-6\tau}-e^{-2\tau})\tau\cos2\sigma].
\end{align}

From (\ref{ap11}) and (\ref{ap12}) we have $\sigma_n e^{-2\tau_n}\to \left|\frac{q(1)}{4\omega}\right|$ as $n\to\infty$. It follows that when $k\in \gamma_{\tau_n}^\pm$,
\begin{align}\label{ap17}
 \sigma e^{-2\tau}=(\sigma_n+t)e^{-2(\tau_n\pm \varepsilon)}\to \left|\frac{q(1)}{4\omega}\right|e^{\mp2\varepsilon} , \quad n\to\infty,
\end{align}
and
\begin{equation}\label{ap18}
\sin 2\sigma=\sin 2(\sigma_n+t)=\pm \cos 2t+o(1) \quad \text{when} \quad \mp\frac{q(1)}{\omega}>0,
\end{equation}
here $t\in[-\varepsilon,\varepsilon]$. Substituting (\ref{ap17}) and (\ref{ap18}) to (\ref{ap16}) we have
\begin{equation}\label{ap19}
  G_1(k)=q(1)^2(1+e^{\mp 4\varepsilon}-2e^{\mp 2\varepsilon}\cos 2t)+o(1),\quad k\in \gamma_{\tau_n}^\pm, \quad n\to\infty.
\end{equation}
Denote $p_1(t):=1+e^{\mp 4\varepsilon}-2e^{\mp 2\varepsilon}\cos 2t$. It is easy to see that $p_1(t)$ has minimum value at $t=0$. Thus, we have
\begin{equation}\label{ap19s}
  G_1(k)\ge q(1)^2(1-e^{\mp 2\varepsilon})^2+o(1),\quad k\in \gamma_{\tau_n}^\pm, \quad n\to\infty.
\end{equation}

Similar to (\ref{ap19}), we can also obtain that for $k\in \gamma_{\sigma_n}^\pm$,
\begin{equation}\label{ap20}
  G_1(k)=q(1)^2(1+e^{- 4t}-2e^{- 2t}\cos 2\varepsilon)+o(1), \quad n\to\infty.
\end{equation}
  where $t\in[-\varepsilon,\varepsilon]$. Denote $p_2(t):=1+e^{- 4t}-2e^{- 2t}\cos 2\varepsilon$.  One can easily obtain that $p_2'(t)>0$ for $e^{-2t}<\cos(2\varepsilon)$ and $p_2'(t)<0$ for $e^{-2t}>\cos(2\varepsilon)$. Thus, it follows that
  \begin{equation}\label{ap20s}
  G_1(k)\ge q(1)^2\sin^2 2\varepsilon+o(1), \quad k\in \gamma_{\sigma_n}^\pm,\quad n\to\infty.
\end{equation}
Together with (\ref{ap19s}) and (\ref{ap20s}), we have proved $c_n>c_\varepsilon>0$.

Now let us pay attention to case $k\in \Gamma_n$. Denote
\begin{equation*}
  C_n:=\inf_{k\in \Gamma_n}G_1(k).
\end{equation*}
Since $g_1(k)$ is odd and satisfies $g_1(k)^*=g_1(-k^*)$, we only need to consider $k\in\Gamma_n\cap\{k:{\rm Re}k\ge0,{\rm Im}k\ge0\}$.

For $k\in\{k:0\le\sigma\le(n+1)\pi,\tau=(n+1)\pi\}$, we have
\begin{equation}\label{ap21}
  G_1(k)= q(1)^2+o(1).
\end{equation}

For $k\in\{k:\sigma=(n+1)\pi,\frac{1}{4}\log n\le\tau\le (n+1)\pi\}$, or equivalently, $\sigma=\sigma_n+(n+1)\pi-\sigma_n$ and $\tau=\tau_n+t$ with $t\in[\frac{1}{4}\log n-\tau_n,(n+1)\pi-\tau_n]$, similar to (\ref{ap20}), and using (\ref{ap11}) and (\ref{ap12}), we have
\begin{equation}\label{ap22}
  G_1(k)=q(1)^2(1+e^{- 4t})+o(1), \quad n\to\infty.
\end{equation}

For $k\in\{k:\sigma=(n+1)\pi,0\le\tau\le \frac{1}{4}\log n\}$, we get
\begin{equation}\label{ap23}
  G_1(k)=16\omega^2(n+1)^2\pi^2e^{-4\tau}[1+o(1)]\to\infty, \quad n\to\infty.
\end{equation}
Together with (\ref{ap21})-(\ref{ap23}), we obtain that $\{C_n\}_{n\ge0}$ has a positive lower bound $C'$. Taking $C_\varepsilon=\min\{c_\varepsilon,C'\}$, we complete the proof.
\end{proof}
Using Lemma \ref{lemma3.2} and (\ref{uo1}) we get that $|g_1(k)|>|8ikD(k)-g_1(k)|$ for $k\in \Gamma_n\cup \gamma_n$ for large $n$. By the Rouch\'{e}'s theorem, we conclude that the number of zeros of the function $kD(k)$ coincides with the number of zeros of $g_1(k)$ inside $\Gamma_n$ or $\gamma_n$. It follows from Lemma \ref{lemma3.0} that all sufficiently large zeros of $D(k)$ are simple.
 By Lemma \ref{lemma3.1}, the zeros of the function $D(k)$ can be numbered as follows: when $\frac{q(1)}{\omega}<0$ denote  the zeros  of $D(k)$ by $\{\pm k_n\}_{n\ge0}\cup\{\pm k_n^*\}_{n\ge1}$; when $\frac{q(1)}{\omega}>0$ denote  the zeros  of $D(k)$ by $\{\pm k_n\}_{n\ge0}\cup\{\pm k_n^*\}_{n\ge0}$. Moreover,
\begin{equation}\label{ap24}
  k_n=\mu_n+\varepsilon_n,\quad \varepsilon_n=o(1),\quad n\to\infty.
\end{equation}
Let us estimate $\varepsilon_n$. Substituting (\ref{ap24}) into (\ref{uo1}), by a direct calculation, we have
\begin{equation}\label{uo0}
  \sin 2\varepsilon_n=\frac{2iq(1)\sin2\mu_n\cos2\varepsilon_n+4i\omega(\mu_n+\varepsilon_n)+\int_0^2K(t)\sin (\mu_n+\varepsilon_n)tdt}{-2iq(1)\cos2\mu_n}.
\end{equation}
Using the asymptotics of $\mu_n$, and noting $g_1(\mu_n)=0$ and $K(\cdot)\in L^2(0,2)$, we get  $\{\varepsilon_n\}\in l^2$. .

If $q\in W_2^2[0,1]$, then we can substitute (\ref{ap24}) into (\ref{t8}), and obtain the more accurate expression of $\varepsilon_n$. Indeed, using (\ref{ap11}) and (\ref{ap12}) we have
\begin{equation}\label{uo3}
  \cos2(\mu_n+\varepsilon_n)=\frac{2n\pi \omega i}{q(1)}[1+O(n^{-1})]= \cos2\mu_n,\; \frac{1}{\mu_n+\varepsilon_n}=\frac{1}{n\pi}\left[1+O\left(\frac{\log n}{n}\right)\right].
\end{equation}
Comparing  (\ref{uo1}) with (\ref{t8}), we have
\begin{equation}\label{uo4}
  \int_0^2 K(t)\sin kt dt=\frac{Q_1i\cos2k}{k}+\frac{Q_2i}{k}+\frac{\int_0^2K_1(t)\cos ktdt}{k},
\end{equation}
where $K_1(\cdot)\in L^2(0,2)$. It follows from (\ref{uo3}) and  (\ref{uo4}) that
\begin{equation}\label{uo5}
  \int_0^2 K(t)\sin (\mu_n+\varepsilon_n)t dt=-\frac{2Q_1\omega}{q(1)}\left(1+\alpha_n\right),
\end{equation}
where $\{\alpha_n\}\in l^2$. Substituting (\ref{uo5}) into (\ref{uo0}), we obtain
\begin{equation}\label{uo6}
  \varepsilon_n=-\frac{Q_1}{4n\pi q(1)}+\frac{\beta_n}{n},\quad \{\beta_n\}\in l^2.
\end{equation}

When $\omega=0$, the asymptotics of zeros of $D(k)$ can be obtained easier (see, e.g., \cite{FY1}).
Let us summarize what we have proved.
\begin{theorem}\label{th4.1}

(i) Assume $\omega:=\int_0^1q(t)dt\ne0$ and $q(1)\ne0$. Denote the transmission eigenvalues of the problem $R(q,h)$  by $\{\lambda_n\}_{n\ge0}\cup\{\lambda_n^*\}_{n\ge1}$ if $\frac{q(1)}{\omega}<0$; by $\{\lambda_n\}_{n\ge0}\cup\{\lambda_n^*\}_{n\ge0}$ if $\frac{q(1)}{\omega}>0$. If $q\in W_2^1[0,1]$ then the sequence $\{\lambda_n\}_{n\ge0}$ has the following asymptotics
 \begin{equation*}
   \sqrt{\lambda_n}=\mu_n+\varepsilon_n,\quad \{\varepsilon_n\}\in l^2.
 \end{equation*}
 If $q\in W_2^2[0,1]$ then the sequence $\{\lambda_n\}_{n\ge0}$ has the following asymptotics
 \begin{equation*}
   \sqrt{\lambda_n}=\mu_n-\frac{Q_1}{4n\pi q(1)}+\frac{\beta_n}{n},\quad \{\beta_n\}\in l^2.
 \end{equation*}
 Here the asymptotics of $\{\mu_n\}$ is given in (\ref{ap11}) and (\ref{ap12}), and $Q_1$ appears in (\ref{t9}).

 (ii)  Assume  $\omega=0$ and $q(1)\ne0$. Then the transmission eigenvalues, denoted  by $\{\lambda_n\}_{n\ge1}$, have the following asymptotics
  \begin{equation}\label{jal}
   \sqrt{\lambda_n}=\frac{n\pi}{2}+\alpha_n,\quad \{\alpha_n\}\in l^2\quad \text{if}\quad q\in W_2^1[0,1];
 \end{equation}
 \begin{equation}\label{jal}
   \sqrt{\lambda_n}=\frac{n\pi}{2}-\frac{Q_1+(-1)^nQ_2}{2q(1)n\pi }+\frac{\kappa_n}{n},\quad \{\kappa_n\}\in l^2\quad \text{if}\quad q\in W_2^2[0,1],
 \end{equation}
where $Q_1$ and $Q_2$ are given in (\ref{t9}).
\end{theorem}

 When $q(1)=0$ and $q'(1)\ne0$, the asymptotics of the transmission eigenvalues can be studied similarly. We provide the following theorem without proving it.
\begin{theorem}\label{th4.2}

Assume  $q\in W_2^2[0,1]$,  $q(1)=0$ and $q'(1)\ne0$.

 (i) If $\omega:=\int_0^1q(t)dt\ne0$, then the transmission eigenvalues of the problem $R(q,h)$, denoted  by $\{\lambda_n\}_{n\ge0}\cup\{\lambda_n^*\}_{n\ge1}$ if $\frac{q'(1)}{\omega}>0$; by $\{\lambda_n\}_{n\ge1}\cup\{\lambda_n^*\}_{n\ge1}$ if $\frac{q'(1)}{\omega}<0$, have the following asymptotics
 \begin{equation*}
   \sqrt{\lambda_n}=\mu_n^1+\varepsilon_n,\quad \{\varepsilon_n\}\in l^2,
 \end{equation*}
 where
 \begin{equation*}
   \mu_n^1=\left\{\begin{split}
             &n\pi+i\left[\log (2n\pi)-\frac{1}{2}\log \left(-\frac{q'(1)}{2\omega}\right)\right]\quad\text{if}\quad\frac{q'(1)}{\omega}<0,\\
           &\left(n+\frac{1}{2}\right)\pi+i\left[\log (2n\pi)-\frac{1}{2}\log \left(\frac{q'(1)}{2\omega}\right)\right]\quad\text{if}\quad\frac{q'(1)}{\omega}>0.
           \end{split}\right.
 \end{equation*}

 (ii) If  $\omega=0$, then the transmission eigenvalues, denoted  by $\{\lambda_n^\pm\}_{n\ge1}$, have the following asymptotics
  \begin{equation*}
 \sqrt{\lambda_n^\pm}=n\pi+s^\pm+\alpha_n,\quad \{\alpha_n\}\in l^2,
    \end{equation*}
    where
 \begin{equation*}
  s^\pm=\left\{\begin{split}
                 &\pm\frac{1}{2}\arccos\left(-\frac{Q_2}{q'(1)}\right)\qquad\qquad\qquad \text{if}\quad \left|\frac{Q_2}{q'(1)}\right|<1,\\
 & -\frac{i}{2}\log \left(-\frac{Q_2}{q'(1)}\pm\sqrt{\frac{Q_2^2}{q'(1)^2}-1}\right)\quad \text{if}\quad \left|\frac{Q_2}{q'(1)}\right|>1,
               \end{split}\right.
 \end{equation*}
and $Q_2$ is given in (\ref{t9}).
 \end{theorem}

 \begin{remark}
From Theorems \ref{th4.1} and \ref{th4.2}, we know that if $q(1)\ne0$ or $q'(1)\ne0$, then all the transmission eigenvalues with sufficiently large modulus are simple algebraically. Under the condition of Theorem \ref{th4.1}(i) or Theorem \ref{th4.2}(i), there  exist at most finitely many real transmission eigenvalues; under the condition of Theorem \ref{th4.1}(ii), there  exist at most finitely many non-real transmission eigenvalues.
 \end{remark}

\section{Inverse spectral analysis}
In this section, we consider the inverse spectral problems. In \cite{TA1}, Aktosun and Papanicolaou have provided the uniqueness theorem and reconstruction algorithm for recovering the potential from all transmission eigenvalues, $h$ and $\gamma$ (appearing in (\ref{dk})). We shall prove that some knowledge of the potential can replace the constant $\gamma$.

\begin{theorem}\label{th5.1}
The potential function $q(x)$ is uniquely determined by $h$ and all the transmission eigenvalues if one of the following conditions holds
\\
(i) $\int_0^1q(x)dx\neq 0$ is known;
\\
(ii) $q\in  L^2(0,1)\cap W_2^{m+1}(1-\delta,1]$ with some $m\ge0$ and $\delta\in(0,1)$, and $q^{(v)}(1)=0$ for $v=\overline{0,m-1}$ and $q^{(m)}(1)\ne0$ is known.
\end{theorem}
\begin{proof}

(i) From (\ref{p2p}) and (\ref{dk}) we see that if $\int_0^1q(x)dx\ne0$ is given, then
\begin{equation}\label{t1}
  \gamma=\frac{1}{2}\int_0^1q(x)dx\left[\lim_{k\to+\infty}E(k)\right]^{-1}.
\end{equation}
Following the uniqueness theorem in \cite{TA1}, we complete the proof for (i).

(ii) It is known that the Jost solution $f(k,x)$ satisfies $f(k,x)=e^{ikx}p(k,x)$ with
\begin{equation}\label{p1}
p(k,x)=1-\frac{1}{2ik}\int_x^1(1-e^{2ik(t-x)})q(t)p(k,t)dt.
\end{equation}
Take $k=i\tau$ with $\tau<0$ in (\ref{p1}), solve it by the method of successive approximations, and get
\begin{equation}\label{p2}
  p_0(i\tau,x)=1,\quad p_{n+1}(i\tau,x)=\frac{1}{2\tau}\int_x^1(1-e^{-2\tau(t-x)})q(t)p_n(i\tau,t)dt,
\end{equation}
\begin{equation}\label{p3}
  p(i\tau, x)=\sum_{n=0}^{\infty}p_n(i\tau, x).
\end{equation}
It is easy to show
\begin{equation*}
|  p_n(i\tau,x)|\le \frac{e^{-2\tau(1-x)}\left(\int_x^1|q(t)|dt\right)^n}{\tau^nn!},\quad n\ge1.
\end{equation*}
Thus the series (\ref{p3}) is absolutely and uniformly convergent. Let us seek the asymptotics of $p(i\tau,0)$ and $p'(i\tau,0)$ as $\tau\to-\infty$ under the condition in Theorem \ref{th5.1}(ii).

Note that for $x\in[0,1-\delta]$ there holds
\begin{align}\label{p4}
\notag \int_x^1e^{-2\tau t}q(t)dt&=\left(\int_x^{1-\delta}+\int_{1-\delta}^1\right)e^{-2\tau t}q(t)dt\\
&=-\frac{q^{(m)}(1)e^{-2\tau}}{(2\tau)^{m+1}}+O(e^{2\tau(\delta -1)}),\;\;\tau\to-\infty,
\end{align}
for $x\in[0,1]$ there holds
\begin{align}\label{p5}
\notag \left|\int_x^1e^{-2\tau t}q(t)dt\right|&\le\left|\int_x^{1-\delta}e^{-2\tau t}q(t)dt\right|+\left|\int_{1-\delta}^1e^{-2\tau t}q(t)dt\right|\\
&\le \frac{\left|q^{(m)}(1)\right|e^{-2\tau}}{|2\tau|^{m+1}}+O(e^{2\tau(\delta -1)}),\;\;\tau\to-\infty.
\end{align}

In each $p_n(i\tau,t)$, there must have the term $\int_s^1e^{-2\tau \xi}q(\xi)d\xi$ (there also has the term $\left(\int_x^1q(s)ds\right)^n/n!$, which, however, is independent of $e^{-2\tau}$). It follows from (\ref{p2}), (\ref{p3}), (\ref{p4}) and (\ref{p5}) that for $x\in[0,1-\delta]$
\begin{equation}\label{p6}
  p(i\tau,x)=p_1(i\tau,x)[1+o(1)]=\frac{q^{(m)}(1)e^{2(x-1)\tau}}{(2\tau)^{m+2}}[1+o(1)],\;\;\tau\to-\infty.
\end{equation}

Using (\ref{p1}) we have
\begin{equation}\label{p7}
  p'(i\tau,x)=-\int_x^1e^{-2\tau(t-x)}q(t)p(i\tau,t)dt,
\end{equation}
which implies from (\ref{p1}) again  that
\begin{equation}\label{p8}
  p'(i\tau,x)=2\tau(p(i\tau ,x)-1)-\int_x^1q(t)p(i\tau, t)dt.
\end{equation}
It follows from (\ref{p5}), (\ref{p6}) and (\ref{p8}) that
\begin{equation}\label{p9}
  p'(i\tau,x)=\frac{q^{(m)}(1)e^{2(x-1)\tau}}{(2\tau)^{m+1}}[1+o(1)],\quad x\in[0,1-\delta],\;\;\tau\to-\infty.
\end{equation}
Thus, we have
\begin{equation}\label{p10}
  f'(i\tau,0)=p'(i\tau,0)-\tau p(i\tau,0)=\frac{q^{(m)}(1)e^{-2\tau}}{2^{m+2}\tau^{m+1}}[1+o(1)],\;\;\tau\to-\infty.
\end{equation}
Together with (\ref{5s}), (\ref{p6s}) and (\ref{p10}), we obtain
\begin{equation}\label{p11}
  F(i\tau)=-\frac{iq^{(m)}(1)e^{-2\tau}}{2^{m+2}\tau^{m+1}}[1+o(1)],\;\;\tau\to-\infty.
\end{equation}
Substituting (\ref{p11}) into (\ref{6s}), and noting $F(i\tau)\to0$ as $\tau\to+\infty$, we get
\begin{equation}\label{p12}
  D(i\tau)=-\frac{1}{4}\cdot\frac{q^{(m)}(1)e^{-2\tau}}{(2\tau)^{m+1}}[1+o(1)],\;\;\tau\to-\infty.
\end{equation}
Therefore, the constant $\gamma$ in (\ref{dk}) can be uniquely recovered by all the transmission eigenvalues and $q^{(m)}(1)$ by the following formula
\begin{equation}\label{p13}
\gamma=\lim_{\tau\to-\infty}\frac{-q^{(m)}(1)e^{-2\tau}}{4E(i\tau)(2\tau)^{m+1}}.
\end{equation}
Following the uniqueness theorem in \cite{TA1}, we complete the proof.
\end{proof}

\begin{remark}
The reconstruction algorithm for recovering the potential $q(x)$ from all transmission eigenvalues, $h$ and the constant $\gamma$ have been given in \cite{TA1}. This fact together with the formulas (\ref{t1}) and (\ref{p13}) implies the reconstruction algorithms for recovering $q(x)$  from all transmission eigenvalues and $h$ with the non-zero value $\int_0^1q(s)ds$ or $q^{(m)}(1)$.
\end{remark}

\section{The Dirichlet case}
 In the Dirichlet case, (i.e., the boundary conditions (\ref{3}) and (\ref{4}) are respectively replaced by $\psi(0)=0$ and  $\psi_0(0)=0$), the characteristic function $D(k)$ satisfies (see \cite{TA1})
\begin{equation}\label{po1}
 D(k)=\frac{1}{2ik}[f(k,0)-f(-k,0)],
\end{equation}
and
\begin{equation}\label{po2}
 D(k)=\frac{\int_0^1q(x)dx}{2k^2}+o\left(\frac{1}{k^2}\right),\quad |k|\to\infty,\quad k\in \mathbb{R}.
\end{equation}
Since $D(k)$ here is also an even and entire function of $k$ of exponential type. By the Hadamard's factorization theorem,
\begin{equation}\label{po3}
  D(k)=\gamma E_0(k),\quad E_0(k):=k^{2s}\prod_{\lambda_n\ne0}\left(1-\frac{k^2}{\lambda_n}\right),
\end{equation}
where $\gamma$ is constant, $s\ge0$ is the multiplicity of the zero eigenvalue, and $\{\lambda_n\}$ are the nonzero transmission eigenvalues.

Note that (\ref{p6}) implies
\begin{equation}\label{p6s}
  f(i\tau,0)=\frac{q^{(m)}(1)e^{-2\tau}}{(2\tau)^{m+2}}[1+o(1)],\;\;\tau\to-\infty.
\end{equation}
It follows from (\ref{po1}) and (\ref{p6s}) that
 \begin{equation}\label{p14}
   D(i\tau)=- \frac{q^{(m)}(1)e^{-2\tau}}{(2\tau)^{m+3}}[1+o(1)],\quad \tau\to-\infty,
 \end{equation}
which implies
\begin{equation}\label{oppo}
\gamma=\lim_{\tau\to-\infty}\frac{-q^{(m)}(1)e^{-2\tau}}{E_0(i\tau)(2\tau)^{m+3}}.
\end{equation}
By (\ref{po2}) and (\ref{po3}), we also get
\begin{equation}\label{oppo1}
\gamma=\frac{1}{2}\int_0^1q(x)dx\left[\lim_{k\to+\infty}k^2E_0(k)\right]^{-1}.
\end{equation}
The formulas (\ref{oppo}) and (\ref{oppo1}) together with the reconstruction algorithms in \cite{TA4} or \cite{XC0} imply the reconstruction algorithm for recovering the potential from all transmission eigenvalues and the non-zero value $\int_0^1q(s)ds$ or $q^{(m)}(1)$.

Using (\ref{po1}) and the properties of the Jost solution in Section 2, we also get the asymptotics of the function $D(k)$. If $q\in W_2^1[0,1]$ then
\begin{equation}\label{dil}
  D(k)=\frac{\int_0^1q(s)ds}{2k^2}-\frac{q(1)\sin 2k}{4k^3}-\frac{1}{k^3}\int_0^2 K_{tt}(0,t)\sin ktdt.
\end{equation}
If $q\in W_2^2[0,1]$ then
\begin{equation}\label{dil1}
  D(k)=\frac{\omega}{2k^2}-\frac{q(1)\sin 2k}{4k^3}+\frac{Q_3\cos 2k-Q_4}{8k^4}-\frac{1}{k^4}\int_0^2 K_{ttt}(0,t)\cos ktdt,
\end{equation}
where
\begin{equation}\label{ona}
  Q_3=-q'(1)+q(1)\omega,\quad Q_4=-q'(0)+q(0)\omega-\int_0^1q^2(s)ds+\frac{\omega^3}{6}.
\end{equation}

The equation (\ref{dil1}) are almost the same as (\ref{t8}). By the same discussion as that in Section 3, we obtain the following theorem.

\begin{theorem}

(i) Assume $\omega:=\int_0^1q(t)dt\ne0$ and $q(1)\ne0$. Denote the Dirichlet transmission eigenvalues by $\{\lambda_n\}_{n\ge1}\cup\{\lambda_n^*\}_{n\ge1}$ if $\frac{q(1)}{\omega}>0$; by $\{\lambda_n\}_{n\ge0}\cup\{\lambda_n^*\}_{n\ge1}$ if $\frac{q(1)}{\omega}<0$. If $q\in W_2^1[0,1]$ then the sequence $\{\lambda_n\}$ has the following asymptotics
 \begin{equation*}
   \sqrt{\lambda_n}=\mu_n+\varepsilon_n,\quad \{\varepsilon_n\}\in l^2.
 \end{equation*}
 If $q\in W_2^2[0,1]$ then the sequence $\{\lambda_n\}$ has the following asymptotics
 \begin{equation*}
   \sqrt{\lambda_n}=\mu_n+\frac{Q_3}{4n\pi q(1)}+\frac{\beta_n}{n},\quad \{\beta_n\}\in l^2.
 \end{equation*}
 Here the asymptotics of $\{\mu_n\}$ is given in (\ref{ap11}) for the case $\frac{q(1)}{\omega}>0$ and (\ref{ap12}) for the case $\frac{q(1)}{\omega}<0$, and $Q_3$ appears in (\ref{ona}).

 (ii)  Assume $\omega=0$ and $q(1)\ne0$. Then the Dirichlet transmission eigenvalues, denoted  by $\{\lambda_n\}_{n\ge1}$, have the following asymptotics
 \begin{equation}\label{jal1}
  \sqrt{\lambda_n}=\frac{(n+1)\pi}{2}+\alpha_n,\quad \{\alpha_n\}\in l^2\quad \text{if}\quad q\in W_2^1[0,1];
 \end{equation}
 \begin{equation}\label{jal1}
   \sqrt{\lambda_n}=\frac{(n+1)\pi}{2}+\frac{Q_3+(-1)^{n}Q_4}{2q(1)n\pi }+\frac{\kappa_n}{n},\quad \{\kappa_n\}\in l^2\quad \text{if}\quad q\in W_2^2[0,1],
 \end{equation}
where $Q_3$ and $Q_4$ are defined in (\ref{ona}).
\end{theorem}

\noindent {\bf Acknowledgments.}
The research work was supported in part by the National Natural Science Foundation of China (11901304) and the Startup Foundation for Introducing Talent of NUIST.

\end{document}